\documentclass[preprint,12pt]{elsarticle}

\usepackage{amssymb, amsmath}
\usepackage{amsthm}
\usepackage{bbm} 

\journal{Nonlinear Analysis: Hybrid Systems}

\newtheorem{theorem}{Theorem}[section]

\newtheorem{corollary}[theorem]{Corollary}
\newtheorem{proposition}[theorem]{Proposition}
\newdefinition{definition}{Definition}
\newdefinition{rmk}{Remark}
\newdefinition{example}{Example} 

\begin{document}

\begin{frontmatter}



\title{Explicit solutions to utility maximization problems in a regime-switching market model via Laplace transforms}


\author[Ocejo]{Adriana Ocejo}
\ead{amonge2@uncc.edu}

\address[Ocejo]{Department of Mathematics and Statistics, University of North Carolina at Charlotte, 9201 University City Blvd., Charlotte, NC, 28223.}

\begin{abstract}
We study the problem of utility maximization from terminal wealth in which an agent optimally builds her portfolio
by investing in a bond and a risky asset.
The asset price dynamics follow a diffusion process with regime-switching coefficients modeled by a continuous-time finite-state Markov chain. 
We consider an investor with a Constant Relative Risk Aversion (CRRA) utility function.
We deduce the associated Hamilton-Jacobi-Bellman equation to construct the solution and the optimal trading strategy
and verify optimality by showing that the value function is the unique constrained viscosity solution of the HJB equation.
By means of a Laplace transform method, we show how to explicitly compute the value function and illustrate the method with the two- and three-states cases.
This method is interesting in its own right and can be adapted in other applications involving hybrid systems and using other types of transforms with basic properties similar to the Laplace transform.
\end{abstract}


\end{frontmatter}



\section{Introduction}
In this paper we study an investment problem of an agent whose portfolio is constructed by investing in a bond and a risky asset,
whose price dynamics follow a diffusion process with regime-switching coefficients,
modeled by an observable continuous-time finite-state Markov chain. 
The agent's objective is to maximize her expected utility from terminal wealth.

Changes of regime in financial markets have been empirically observed and may be due, for instance, to sudden changes in the economy or major political events.
Regime-switching processes were initially proposed by Hamilton, 
who studied the effect of incorporating shifts in the parameters of a discrete-time model, via an unobserved discrete time two-state Markov chain,
when analyzing yields on government bonds \cite{Ham89}.
Since then, 
many empirical studies have argued that regime-switching modeling can help to better predict market prices behavior.
More recently, Pereiro and Gonz\'alez-Rozada \cite{Pereiro} analyzed the market index of a sample of stock markets worldwide to test for the presence of regimes.
They concluded that $68\%$ and $37\%$ of the emerging and developed stock markets, respectively, show the existence of regime-switching, including the SPX in USA.

Portfolio optimization problems in continuous-time date back to the works of Merton \cite{Merton69}, \cite{Merton71},
who proposed that the market risk is driven by a Brownian motion.
In the context of regime-switching dynamics, where an auxiliary Markov process dictates the market regime,
the utility maximization problem from terminal wealth has been studied under different assumptions on how regime information is available to the agent. 

For partially observable regimes in which the Markov chain is hidden or not observed directly, 
see Sass and Haussmann \cite{Sass}, Nagai and Runggaldier \cite{Nagai} and references therein.
The authors argue that explicit analytical solutions are very difficult to obtain and the optimal strategies and value function have to be determined numerically.

In the fully observable case and for an investor with logarithmic or power utilities,
Capponi and Figueroa-L\'opez \cite{Capponi&Figuroa} and Fu et al. \cite{Fuetal}
consider a portfolio that contains, besides a risk-free bond and a risky stock, an extra term.
In \cite{Capponi&Figuroa}, the authors take into account default risk and incorporate a defaultable bond into the portfolio.
They allow regime dependent short rate, drift, volatility and default intensities.
By separating the problem into pre- and post default optimization subproblems, they provide the associated verification theorems for each subproblem assuming that the associated HJB equation has a smooth solution, and construct the value functions as the solution of coupled linear systems of ordinary differential equations.
Fu et al. \cite{Fuetal} consider a portfolio that also contains an option written on the stock. They approximate the value function as the limit of a sequence of value functions of auxiliary problems.
More related to our paper with a portfolio built with a risk-free bond and a risky asset only,
Zhang and Yin \cite{Zhang&Yin} consider a fairly general setup for the utility function, including the power, logarithmic and exponential functions.
However, due to this generality, the HJB equation of the associated control problem is too difficult to solve explicitly. Then, they opt to tackle the problem using a singular-perturbation approach and successfully obtain \textit{near-optimal} allocation strategies.


In this paper, we assume that the short rate, as well as the rate of return and volatility of the risky-asset depend on the market regime.
The state of the market is modeled by an observable continuous-time and finite-state Markov chain.
We allow the cash amount to be invested in the risky asset to be unbounded.
Due to this relaxed assumption on the control set, the HJB equation associated to the maximization problem is of
degenerate parabolic type. Therefore, a smooth solution $V$ cannot be assumed to exist and the classical verification result based on an application of It\^o's formula is not possible.

The notion of viscosity solutions has been successfully applied in many contexts when a smooth solution to a PDE equation is not expected to exist.
For instance, in the case of stochastic controlled problems, a non-exhaustive list includes the works of Lions \cite{Lions}, Duffie and Zariphopoulou \cite{Duffie&Zarip}, Duffie et al. \cite{Duffie_etal}, 
Kounta \cite{Kounta}, and others. 
In the case of optimal stopping problems Pemy and Zhang \cite{Pemy&Zhang}, Li \cite{Li2016}, Bian et al. \cite{Bianetal}.
For a general overview of the theory of viscosity solutions of second order PDEs we refer to Crandall, Ishii and Lions \cite{Crandalletal},
and for their connection to stochastic optimization problems we refer to Fleming and Soner \cite{Fleming&Soner} . 
We shall show that $V$ is the unique \textit{constrained} viscosity solution of the associated HJB equation in an appropriate class of functions, in a sense to be specified later.

One of the main contributions of this paper is that
we present a simple methodology to explicitly compute the value function based on inverse Laplace transforms, which is a new idea in the literature on portfolio optimization problems.
We think that this idea can be easily applied to other optimization problems involving regime-switching diffusions and is interesting in its own right.


The rest of the paper is organized as follows. The model dynamics and problem formulation is presented in Section \ref{sec:model}.
In Section \ref{Construction}, we heuristically construct a candidate trading (portfolio allocation) strategy and value function by martingale and dynamic programming arguments.
In Section \ref{sec:ver}, we show that $V$ is the unique \textit{constrained} viscosity solution of the associated HJB equation, in a sense to be specified later, and
a verification that the candidate solution satisfies all the necessary conditions, implying that it coincides with the value function.
We also obtain optimal allocation strategies in the form of feedback controls.
The Laplace transform method to compute the value function is presented in Section \ref{sec:Laplace} and
some numerical experiments are provided in Section \ref{sec:numerical}. The final section summarizes the main results.

\section{Model dynamics and problem setup} \label{sec:model}

Let $(\Omega,\mathcal{F},\mathbb{P})$ be a probability space which supports a Brownian motion $B=(B_t)_{t\geq 0}$
and a continuous-time Markov chain $Y=(Y_t)_{t\geq 0}$ with finite state space $\mathcal{M}=\{1,2,\ldots,m\}$ and generator $Q=(q_{i j})_{m\times m}$ which satisfies
\[
q_{i j}\geq 0 \quad \mbox{for $i\neq j$}, \qquad \sum_{j\in \mathcal{M}} q_{i j}=0, \qquad q_i:=-q_{i i}>0.
\]
We denote by $\mathbb{F}=(\mathcal{F}_t)_{t\geq 0}$ the $\mathbb{P}$-augmentation of the filtration generated by $B$ and $Y$.
Lemma 2.5 in \cite{Jacka&Mija} ensures that $B$ and $Y$ are independent. 

Let $P=(P_t)_{t\geq 0}$ and $S=(S_t)_{t\geq 0}$ denote the price of the bond and the risky asset, respectively.
We assume that these processes satisfy the Markov-modulated dynamics
\[
\begin{aligned}
dP_t    & = r(Y_t) P_t\,dt,   \\
dS_t    & = \mu(Y_t)S_t\,dt+\sigma(Y_t)S_t\,dB_t
\end{aligned}
\]
where $r(i)>0$, $\mu(i)>0$ and $\sigma(i)>0$ denote the risk-free interest rate, the rate of return of the asset, and the volatility at regime $i$, respectively.

At every time $s\in[t,T]$, with $0\leq t\leq T$, the agent chooses the cash amount $\theta_s$ to be invested in the asset,
and allocates the rest of his wealth $X_s-\theta_s$ in the bond.
We do not impose a portfolio constraint, in particular short-selling is allowed ($\theta_s<0$).
Assuming the self-financing condition, his wealth process $X=(X_s)_{s\geq t}$ satisfies the {\it controlled} stochastic differential equation
\begin{equation} \label{eq:wealth_dyn}
dX_s= \theta_s[\mu(Y_s)-r(Y_s)]ds+ r(Y_s)X_s\,ds+\theta_s \sigma(Y_s)\,dB_s,
\end{equation}
with initial endowment $X_t=x\geq 0$,  
and impose that wealth stays non-negative
\[
X_s\geq 0, \qquad \mbox{for all $s\geq t$ a.s. }
\]
We are interested in {\it feedback type} controls $\theta=(\theta_s)_{t\leq s\leq T}$ of the form
$\theta_s=\bar{\theta}(X_s,s,Y_s)$ for some function $\bar{\theta}(x,s,i):[0,\infty)\times [0,T]\times\mathcal{M}\mapsto \mathbb{R}$
such that for each $i$, $\bar{\theta}(\cdot,\cdot,i)$ is Borel measurable.
Sufficient conditions to ensure the existence and (pathwise) uniqueness of a solution of a regime-switching stochastic differential equations (c.f. Theorem 4.1 in \cite{Fleming&Rishel},\cite{Yin&Zhu}) call for linear growth and Lipschitz conditions on the coefficients. In the case of (\ref{eq:wealth_dyn}) with feedback controls, we shall assume that there exist constants $K_1,K_2>0$ such that for each $i$,
\begin{equation} \label{eq:coeff_cond}
|\bar{\theta}(x,t,i)|\leq K_1(1+|x|), \qquad \mbox{and} \qquad
|\bar{\theta}(x,t,i)-\bar{\theta}(y,t,i)|\leq K_2|x-y|.
\end{equation}

\begin{definition}
Given $X_t=x, Y_t=i$, a trading strategy $\theta=(\theta_s)_{t\leq s\leq T}$ is said to be \textit{admissible} if
it is of feedback type satisfying (\ref{eq:coeff_cond}),
the integrability condition
\begin{equation} \label{eq:admissibility_cond}
\mathbb{E}\left[\int_t^T \theta_s^2 ds \mid \mathcal{F}_t \right]<\infty
\end{equation}
holds, and the unique solution $X_s$ to the SDE in (\ref{eq:wealth_dyn}) using this $\theta$ satisfies the state constraint $X_s\geq 0$ for all $s\geq t$ a.s.
We denote the set of all admissible trading strategies by $\mathcal{A}(x,i,t)$. 
\end{definition}

Note that the integrability condition (\ref{eq:admissibility_cond}) ensures that the stochastic integral in (\ref{eq:wealth_dyn}) is well-defined.

The agent's risk preferences are represented by a utility function $U:[0,\infty)\rightarrow \mathbb{R}$ which is strictly increasing, strictly concave and twice continuously differentiable.
The associated \textit{Arrow-Pratt coefficient of absolute risk aversion}
$
A(x)=-U^{''}(x)/U'(x)
$
is interpreted as a measure of aversion to risk \cite{Pratt}. 
In this paper, we consider utility functions with constant relative risk aversion (CRRA), that is, for which
\begin{equation}\label{eq:averParameter}
xA(x)=a, \qquad x>0,
\end{equation}
where $a$ is a positive constant.
We specialize to the power utility $U(x)=\frac{1}{\gamma}x^\gamma$, with risk-aversion parameter $\gamma<1,\gamma\neq 0$. Here $a=1-\gamma$.


The agent's objective is to maximize the expected utility of wealth
\[
J(x,t,i;\theta)=\mathbb{E}[U(X_T)\mid X_t=x, Y_t=i]
\]
over all $\theta \in \mathcal{A}(x,t,i)$ and find the value function
\begin{equation} \label{OriginalProblem}
V(x,t,i)=\sup_{\theta \in \mathcal{A}(x,t,i)}J(x,t,i;\theta)
\end{equation}
for every $x\geq 0$, $t\in[0,T)$, $i\in \mathcal{M}$, with terminal condition $V(x,T,i)=U(x)$.

We will find an explicit solution to this problem. 
In other words, we will find an optimal trading strategy $\theta^*$ which attains the maximal expected utility value
\[
V(x,t,i)=J(x,t,i;\theta^*)
\]
and compute the value function $V$ explicitly.

\section{Construction of the solution} \label{Construction}

By classical stochastic control arguments we know that $V(X_t,t,Y_t)$ is a supermartingale under an arbitrary strategy $\theta$ and a martingale under the optimal strategy $\theta^*$. This leads to the Hamilton-Jacobi-Bellman (HJB) equation associated with the stochastic controlled problem (\ref{OriginalProblem}), namely,
\begin{equation} \label{eq:HJBunconstrained}
\sup_{\theta\in \mathbb{R}} \mathbb{L}^\theta V(x,t,i)  =0,    \qquad  (x,t,i)\in [0,\infty) \times [0,T)\times \mathcal{M}
\end{equation}
with terminal condition $V(x,T,i) = U(x)$ for $x\in [0,\infty)$,
where
\[
\mathbb{L}^\theta V(x,t,i)= V_t +\frac{1}{2}\theta^2 \sigma^2(i)V_{xx}
    + \theta(\mu(i)-r(i))V_x + r(i)x V_x+ QV(x,t,\cdot)(i).
\]
Here, we use the notation $f_x=\frac{\partial f}{\partial x}$, $f_{xx}=\frac{\partial^2 f}{\partial x^2}$, $f_t=\frac{\partial f}{\partial t}$ and
\[
Qf(x,t,\cdot)(i)=\sum_{j\neq i}[f(x,t,j)-f(x,t,i)]q_{ij}.
\]

It is important to remark that 
the HJB equation (\ref{eq:HJBunconstrained}) is a degenerate, second-order parabolic equation and therefore it may not have a smooth solution, or a solution may not even exist in the classical sense. For this reason, a solution to (\ref{eq:HJBunconstrained}) will be understood as a {\it viscosity solution} in a sense to be specified later in the next section.
Below, we heuristically construct a candidate optimal strategy $\theta^*$ and value function $V$ solving (\ref{eq:HJBunconstrained}) under the assumption that $V$ is sufficiently smooth. The rigorous verification of optimality and that such candidate is indeed the value function is presented in the next section.

Formally, since $\mathbb{L}^\theta$ is quadratic in $\theta$, and assuming that $V_{xx}<0$, the maximum is attained by
\begin{equation} \label{eq:candidatestrat}
\theta_t^*=-\frac{[\mu(Y_t)-r(Y_t)]V_x}{\sigma^2(Y_t)V_{xx}}
\end{equation}
and $V$ in (\ref{eq:HJBunconstrained}) solves the coupled nonlinear partial differential equation (PDE)
\begin{equation}\label{eq:HJBconstrained}
\begin{aligned}
V_t+r(i)xV_x+QV(x,t\cdot)(i)-\frac{(\mu(i)-r(i))^2V^2_x}{2\sigma^2(i)V_{xx}}  & =0 \\
                                                                    V(x,T,i)& = U(x).
\end{aligned}
\end{equation}
The goal is to find an explicit solution $V(x,t,i)$ to this system. To do so, we specify the utility function $U(x)$ further.

Consider the power utility $U(x)=\frac{x^\gamma}{\gamma}$, with $\gamma<1$ and $\gamma\neq 0$.
The linearity of the wealth $X_s$ and strategy $\theta_s$ in (\ref{eq:wealth_dyn}),
the property $U(xy)=U(x)y^\gamma$ of the utility function,
and the Markov property,
suggest the following ansatz for the value function
\begin{equation} \label{eq:ansatz}
V(x,t,i)=U(x) g(i,T-t)
\end{equation}
where $g(i,0)=1$.
Substitution of this expression into the HJB equation (\ref{eq:HJBconstrained}) and setting $g\equiv g(i,T-t)$, yields
\[
\begin{aligned}
g_t-r(i)x\frac{U'(x)}{U(x)}g-Qg(\cdot,T-t)(i)+\frac{(\mu(i)-r(i))^2}{2\sigma^2(i)}\frac{[U'(x)]^2}{U(x)U^{''}(x)}g  & =0 \\
                                                g(i,0)& =1.
\end{aligned}
\]

It is easy to see that $-x\frac{U^{''}(x)}{U'(x)}=1-\gamma$ and $x\frac{U'(x)}{U(x)}=\gamma$.
This gives the following coupled PDE equation, which does not depend on $x$,
\[
\begin{aligned}
g_t-Qg(\cdot,T-t)(i)-\gamma\left[\frac{(\mu(i)-r(i))^2}{2(1-\gamma)\sigma^2(i)}+r(i)\right]g  & =0 \\
                                                g(i,0)& =1.
\end{aligned}
\]

This is the regime-switching version of a classical Cauchy problem (see \cite{Baran}), and the stochastic representation of the solution is given by
{\small
\begin{equation} \label{eq:F-Krep}
g(i,T-t)=\mathbb{E}\left[\exp\left\{\int_0^{T-t} \gamma\left[\frac{(\mu(Y_u)-r(Y_u))^2}{2(1-\gamma)\sigma^2(Y_u)}+ r(Y_u)\right]du \right\} \mid Y_0=i \right].
\end{equation}
}

In Section \ref{sec:Laplace}, we shall show the general method to compute the value function using Laplace transforms.


\section{Verification} \label{sec:ver}

Typically, a viscosity solution is defined on an open subset of the state space, in this case in $(0,\infty)$, and such solution is not a priori defined on the actual state space $[0,\infty)$ (recall $X_s\geq 0$).
To incorporate state constraints (such as that at the boundary level $x=0$), Soner \cite{Soner} 
introduced the concept of {\it constrained} viscosity solution and the idea has been adapted in \cite{Duffie&Zarip}, \cite{Duffie_etal}, and \cite{Zarip94}
for second-order operators.

In this section we first show the existence of a constrained viscosity solution.
Then state that $V$ is the unique constrained viscosity solution in the class of concave functions in $x$ with the terminal condition $V(x,i,T)=U(x)$
and that $V$ is sufficiently smooth.
Finally, we verify that the candidate function constructed in the previous section coincides with the value function.

Let $\mathcal{O}=(0,\infty)\times [0,T)$, so that the closure $\bar{\mathcal{O}}=[0,\infty)\times [0,T]$.
Consider the function
\[
H: \bar{\mathcal{O}}\times \mathcal{M}\times \mathbb{R}\times \mathbb{R}\mapsto \mathbb{R}
\]
defined by
\[     
H(x,t,i,p,A) := \min_{\theta\in\mathbb{R}}\left\{ -\frac{1}{2}\theta^2 \sigma^2(i)A-\theta(\mu(i)-r(i))p\right  \} - r(i)x p
\]
For each fixed $i\in \mathcal{M}$, the function $H$ is continuous. 
Moreover, $H$ satisfies the property of degenerate ellipticity, namely
\[
H(x,t,i,p,A+B)\leq H(x,t,i,p,A),\qquad \mbox{if $B\geq 0$}.
\]
Now consider the function
\[
F(x,t,i,v,v_t,v_x,v_{xx}) := H(x,t,i,v_x,v_{xx}) -v_t-Qv(x,t,\cdot)(i).
\]
Formally, the HJB equation associated with the value function in (\ref{OriginalProblem}) can be written as
\begin{equation} \label{eq:viscosityHJB}
F(x,t,i,v,v_t,v_x,v_{xx})=0. 
\end{equation}
We now state the definition of a constrained viscosity solution of (\ref{eq:viscosityHJB}). 
We follow \cite{Duffie&Zarip} and \cite{Duffie_etal} closely and give the natural modification of the definition of constrained viscosity solution in our setting.

\begin{definition} \label{def:viscosity}
A function $v:\bar{\mathcal{O}}\times \mathcal{M}\mapsto \mathbb{R}$ is a {\it viscosity solution} of (\ref{eq:viscosityHJB}) in $\mathcal{O}$
if for each $i\in \mathcal{M}$,
$v(\cdot,\cdot,i)$ is continuous and the following holds:
\begin{itemize}
\item[(i)] $v(\cdot,\cdot,i)$ is a {\it viscosity supersolution of (\ref{eq:viscosityHJB}) in $\mathcal{O}$}, that is, if for any test function $\phi\in C^2(\bar{\mathcal{O}})$ 
and any local minimum $(x_0,t_0)\in \mathcal{O}$ of $v-\phi$ it follows that
\[
F(x_0,t_0,i,v(x_0,t_0,i),\phi_t(x_0,t_0),\phi_x(x_0,t_0),\phi_{xx}(x_0,t_0))\geq  0.
\]
\item[(ii)] $v(\cdot,\cdot,i)$ is a {\it viscosity subsolution of (\ref{eq:viscosityHJB}) in $\mathcal{O}$}, that is, if for any test function $\phi\in C^2(\bar{\mathcal{O}})$ 
and any local maximum $(x_0,t_0)\in\mathcal{O}$ of $v-\phi$ it follows that
\[
F(x_0,t_0,i,v(x_0,t_0,i),\phi_t(x_0,t_0),\phi_x(x_0,t_0),\phi_{xx}(x_0,t_0))\leq 0.
\]
\end{itemize}
\end{definition}

\begin{definition}
A function $v:\bar{\mathcal{O}}\times \mathcal{M}\mapsto \mathbb{R}$ is a {\it constrained viscosity solution} of (\ref{eq:viscosityHJB}) on $\bar{\mathcal{O}}\times \mathcal{M}$
if for each $i$, $v(\cdot,\cdot,i)$ is a viscosity supersolution (resp. subsolution) of (\ref{eq:viscosityHJB}) in $\mathcal{O}$ (resp. on $\bar{\mathcal{O}}$).
\end{definition}

Throughout the section, we shall use $\mathbb{E}_{x,t,i}[\cdot]$ to denote the conditional expectation $\mathbb{E}[\,\cdot\mid X_{t}=x,Y_{t}=i]$.

\subsection{Analytical properties of the value function}

In this subsection we present some properties of the value function which are necessary to carry over with the existence and uniqueness of constrained viscosity solutions in the appropriate class.

\begin{proposition}
The value function $V$ satisfies that $|V(x,t,i)|\leq O(x^\gamma)$.
\end{proposition}
\begin{proof}
Given that $U(x)=\frac{x^\gamma}{\gamma}$, and the admissible controls satisfy a linear growth condition, the result follows from estimates of the moments of regime-switching diffusions (see the bound $N$ in Appendix A in \cite{Jacka&Ocejo}).
\end{proof}

\begin{proposition} \label{prop:concave}
The function $V(\cdot,t,i)$ is concave and non-decreasing on $[0,\infty)$, for each fixed $t\in[0,T),i\in \mathcal{M}$.
\end{proposition}
\begin{proof}
Let $x_1,x_2\geq 0$, $t\in [0,T)$, $i\in \mathcal{M}$.
Due to the linear dependence of the wealth dynamics (\ref{eq:wealth_dyn}) on the control $\theta$ and initial condition $x$, it follows that
for any $\theta_1\in\mathcal{A}(x_1,t,i)$ and $\theta_2\in\mathcal{A}(x_2,t,i)$ and fixed $\lambda\in(0,1)$
\[
\bar{\theta}:=\lambda \theta_1+(1-\lambda)\theta_2 \in \mathcal{A}(\lambda x_1+(1-\lambda)x_2,t,i).
\]
For $\epsilon>0$, suppose that $\theta_1,\theta_2$ are $\epsilon-$optimal controls for $V(x_1,t,i)$ and $V(x_2,t,i)$, respectively.
Using the concavity of the utility function $U$ we obtain
\[
\begin{split}
V(\lambda x_1+(1-\lambda)x_2,t,i)   & \geq J(\lambda x_1+(1-\lambda)x_2,t,i;\bar{\theta})\\
                                    & \geq  \lambda\,J(x_1,t,i;\theta_1)+(1-\lambda)\,J(x_2,t,i;\theta_2)\\
                                    & \geq  \lambda\,V(x_1,t,i)+(1-\lambda)\,V(x_2,t,i)-\epsilon
\end{split}
\]
which proves the concavity of $V$ in the parameter $x$ since $\epsilon>0$ can be made arbitrarily small.

Now suppose that $x_1\leq x_2$.
Using standard path comparison theorems (see Section IX.3 \cite{Rev-Yor}) of solutions of stochastic differential equations,
we have that $X^{t,x_1}_s\leq X^{t,x_2}_s$ for all $s\geq t$ a.s. where $X^{t,x_1}_t=x_1$ and $X^{t,x_2}_t=x_2$, respectively. This yields
$\mathcal{A}(x_1,t,i)\subset \mathcal{A}(x_2,t,i)$, which implies that $V$ is non-decreasing as a function of $x$.
\end{proof}

Proposition \ref{prop:concave} along with the continuity of $V(\cdot,t,i)$ at $x=0$ (c.f. \cite{Karatzas_etal}) yield that $V(\cdot,t,i)$ is Lipschitz continuous in $[0,\infty)$.

\pagebreak
\begin{proposition}  \label{prop:t_cont}
The function $V(x,\cdot,i)$ is $1/2-$H\"older continuous in $[0,T]$ uniformly over a neighborhood of $x$, for each fixed $x\in[0,\infty), i\in \mathcal{M}$.
\end{proposition}
\begin{proof}
Let $0\leq t_1<t_2\leq T$. By the Dynamic Programming Principle,
\[
\begin{aligned}
|V(x,t_1,i)-V(x,t_2,i)| & = \left| \sup_{\theta\in \mathcal{A}(x,t_1,i)}\mathbb{E}_{x,t_1,i}[V(X_{t_2},t_2,Y_{t_2})]-V(x,t_2,i)\right| \\
                        & \leq \sup_{\theta\in \mathcal{A}(x,t_1,i)} \mathbb{E}_{x,t_1,i}[\,|V(X_{t_2},t_2,Y_{t_2})-V(x,t_2,i)|\,] \\
                        & \leq N'\sup_{\theta\in \mathcal{A}(x,t_1,i)} \mathbb{E}_{x,t_1,i}[\,|X_{t_2}-x|\,]
\end{aligned}
\]
for some constant $N'>0$, where the last inequality is due to the Lipschitz continuity of $V(\cdot,t_2,i)$.
Given that the admissible controls in $\mathcal{A}(x,t_1,i)$ satisfy a linear growth condition, it follows that
\[
|V(x,t_1,i)-V(x,t_2,i)|\leq N|t_1-t_2|^{1/2}
\]
where $N$ depends on $N',x,i,T$ and it is continuous as a function of $x$ (see the proof of Proposition 2.1 in \cite{Jacka&Ocejo}), which concludes the proof.
\end{proof}

\subsection{Existence of constrained viscosity solutions}


\begin{proposition} \label{prop:existence}
The value function $V:\bar{\mathcal{O}}\times \mathcal{M}\mapsto \mathbb{R}$ is a constrained viscosity solution of (\ref{eq:viscosityHJB}) on $\bar{\mathcal{O}}\times \mathcal{M}$.  
\end{proposition}
\begin{proof}
Fix $i\in \mathcal{M}$.

\underline{Step 1.} We first show that $V(\cdot,\cdot,i)$ is a viscosity supersolution of (\ref{eq:viscosityHJB}) in $\mathcal{O}$.
Let $\phi\in C^2(\bar{\mathcal{O}})$ and $(x_0,t_0)\in\mathcal{O}$ be a minimum of $V-\phi$ in a neighborhood $N(x_0,t_0)\subset \mathcal{O}$.
We want to prove that
\begin{equation}\label{eq:supersolHJB}
\begin{split}
0 &\geq \max_{\theta\in\mathbb{R}}\left\{ \frac{1}{2}\theta^2 \sigma^2(i)\phi_{xx}(x_0,t_0)+\theta(\mu(i)-r(i))\phi_x(x_0,t_0)\right  \} \\
& \hspace{1cm}+ r(i)x_0 \phi_x(x_0,t_0) +\phi_t(x_0,t_0)+QV(x_0,t_0,\cdot)(i).
\end{split}
\end{equation}

Let $\theta \in \mathcal{A}(x_0,t_0,i)$ be a constant control, that is $\theta_s\equiv \theta\in \mathbb{R}$ for all $s\geq t_0$,
and let $X$ be the unique solution of (\ref{eq:wealth_dyn}) using this $\theta$ which starts at $X_{t_0}=x_0$ and suppose that the Markov chain starts at $Y_{t_0}=i$.

Define the stopping time
\[
\tau:=\inf\{s\geq t_0:\,(X_s,s)\notin N(x_0,t_0)\}\wedge \inf\{s\geq t_0:\,Y_s\neq Y_{t_0}\}.
\]
We also define the function
\[
\Phi(x,t,j):=
\begin{cases}
\phi(x,t) & \mbox{if $j=i$}, \\
V(x,t,j)    & \mbox{if $j\neq i$}.
\end{cases}
\]

On the one hand, the dynamic programming principle implies that
\begin{equation} \label{eq:DDP}
V(x_0,t_0,i)\geq \mathbb{E}_{x_0,t_0,i}[V(X_{s\wedge \tau},s\wedge \tau, Y_{s\wedge \tau})].
\end{equation}

On the other hand, the generalized It\^o's formula for regime-switching diffusions (see e.g. \cite{Baran}) applied to $\Phi(X_s,s,Y_s)$ gives (we omit the dependence of $\Phi$ on the parameters for simplicity)
\[
\begin{split}
d\Phi(X_s,s,Y_s)   & =\Phi_t\,ds+\Phi_x[\theta(\mu(Y_s)-r(Y_s))+r(Y_s)X_s]ds + \frac{1}{2}\Phi_{xx}\,\theta^2 \sigma^2(Y_s)\,ds \\
                & \qquad + Q\Phi(X_s,s,\cdot)(i)ds +dM_s
\end{split}
\]
where $M_s=\int_{t_0}^s \Phi_x(X_u,u,Y_u)\,\sigma(Y_u)\theta \,dB_u$.
By assumption, $\phi_x$ is continuous and so $\Phi_x(X_u,u,Y_u)$ is bounded for $u\in [t_0,s\wedge \tau]$.
Moreover, $\sigma(Y_u)$ is also bounded for any $u$.
Then the process $\{M_{s\wedge \tau}\}_{t_0\leq s\leq T}$ is a square integrable martingale (see e.g. \cite{Rev-Yor}) 
and we obtain Dynkin's formula

{\small
\begin{equation} \label{eq:DynkinForm}
\begin{aligned}
& \mathbb{E}_{x_0,t_0,i}[\Phi(X_{s\wedge \tau},s\wedge \tau, Y_{s\wedge \tau})] = \mathbb{E}_{x_0,t_0,i}[\phi(X_{s\wedge \tau},s\wedge \tau)]=\phi(x_0,t_0) \\
&+\mathbb{E}_{x_0,t_0,i}\left[\int_{t_0}^{s\wedge \tau}\left\{ \frac{\theta^2 \sigma^2(i)}{2}\phi_{xx}+[\theta(\mu(i)-r(i))+r(i)X_u]\phi_x +\phi_t+QV(X_u,u,\cdot)(i)\right\}du\right]
\end{aligned}
\end{equation}
}
where we omitted the dependence of the derivatives of $\phi$ on $(X_u,u)$.

Combining (\ref{eq:DDP}) and (\ref{eq:DynkinForm}), together with the fact that $(x_0,t_0)$ is a minimum in $N(x_0,t_0)$, we get
{\small
\begin{equation}
0\geq \mathbb{E}_{x_0,t_0,i}\left[\int_{t_0}^{s\wedge \tau}\left\{ \frac{1}{2}\theta^2 \sigma^2(i)\phi_{xx}+[\theta(\mu(i)-r(i))+r(i)X_u]\phi_x +\phi_t+QV(X_u,u,\cdot)(i)\right\}du\right].
\end{equation}
}

Dividing by $s-t_0$ and using that $\phi_x, \phi_{xx},\phi_t$ and $V(\cdot,\cdot,j)$ for any $j$ are continuous at $(x_0,t_0)$, as well as the continuity of the paths of $X$,
upon letting $s\rightarrow t_0$ we obtain
\begin{equation}
0\geq \frac{1}{2}\theta^2 \sigma^2(i)\phi_{xx}(x_0,t_0)+[\theta(\mu(i)-r(i))+r(i)x_0]\phi_x(x_0,t_0) +\phi_t(x_0,t_0)+QV(x_0,t_0,\cdot)(i)
\end{equation}
and this is true for any $\theta\in \mathbb{R}$, thus the claim in (\ref{eq:supersolHJB}) follows.

\smallskip
\underline{Step 2.} We now show that $V(\cdot,\cdot,i)$ is a viscosity subsolution of (\ref{eq:viscosityHJB}) on $\bar{\mathcal{O}}$.
Let $\phi\in C^2(\bar{\mathcal{O}})$ and $(x_0,t_0)\in\bar{\mathcal{O}}$ a local maximum of $V-\phi$.
We want to show that (\ref{eq:supersolHJB}) holds with the reversed inequality. However, due to the lack of compactness of the control space ($\theta\in\mathbb{R}$),
some technical difficulties arise. We use a common trick, the stability properties of viscosity solutions (c.f. \cite{Duffie&Zarip}), to approximate $V$ by a sequence of value functions $V^n$ with compact control space which are viscosity subsolutions of a modified HJB equation on $\bar{\mathcal{O}}$. Namely, define
\[
V^n(x,t,i):=\sup_{\theta\in\mathcal{A}^n(x,t,i)} J(x,t,i;\theta)
\]
where $\mathcal{A}^n(x,t,i)=\{\theta\in \mathcal{A}(x,t,i):\, |\theta_s|\leq n,\; a.s.\;\forall s\geq t\}$.
The need for a compact control space will be apparent below.

It is enough to show that $V^n$ is a viscosity subsolution of the modified HJB equation $F_n(x,t,i,v,v_t,v_x,v_{xx})=0$ on $\bar{\mathcal{O}}$,
where $F_n$ is as $F$ in (\ref{eq:viscosityHJB}) with $H$ replaced by
\[
H_n(x,t,i,p,A) := \min_{|\theta|\leq n}\left\{ -\frac{1}{2}\theta^2 \sigma^2(i)A-\theta(\mu(i)-r(i))p\right  \} - r(i)x p.
\]
By stability properties, we have that if $V^n(\cdot,\cdot,i) \rightarrow V(\cdot,\cdot,i)$ locally uniformly on $\bar{\mathcal{O}}$
then $V$ is a subsolution of (\ref{eq:viscosityHJB}) on $\bar{\mathcal{O}}$ (c.f. Theorem 4.1 in \cite{Barles}).

We proceed to prove that the modified value function $V^n$ is a viscosity subsolution of $F_n(x,t,i,v,v_t,v_x,v_{xx})=0$ on the closed domain $\bar{\mathcal{O}}$.
To this end, we need to show that
for $\phi\in C^2(\bar{\mathcal{O}})$ and $(x_0,t_0)\in\bar{\mathcal{O}}$ a local maximum of $V^n-\phi$,
\begin{equation}\label{eq:subsolHJB}
\begin{split}
0 &\leq \max_{|\theta|\leq n}\left\{ \frac{1}{2}\theta^2 \sigma^2(i)\phi_{xx}(x_0,t_0)+\theta(\mu(i)-r(i))\phi_x(x_0,t_0)\right  \} \\
& \hspace{1cm}+ r(i)x_0 \phi_x(x_0,t_0) +\phi_t(x_0,t_0)+QV^n(x_0,t_0,\cdot)(i).
\end{split}
\end{equation}
By contradiction, assume that (\ref{eq:subsolHJB}) is not true.
Then there is a test function $\phi\in C^2(\bar{\mathcal{O}})$ and a local maximum $(x_0,t_0)\in \bar{\mathcal{O}}$ of $V^n-\phi$ such that
the negative of the right-hand side in (\ref{eq:subsolHJB}) is strictly positive.
Using the continuity of $F_n$, 
given $\epsilon>0$ there exists a neighborhood $N(x_0,t_0)\subset \bar{\mathcal{O}}$ such that
\begin{equation} \label{eq:eps_subopt}
\begin{split}
\epsilon < &-\max_{|\theta|\leq n}\left\{ \frac{1}{2}\theta^2 \sigma^2(i)\phi_{xx}(x,t)+\theta(\mu(i)-r(i))\phi_x(x,t)\right  \} \\
& \hspace{1cm}- r(i)x \phi_x(x,t) -\phi_t(x,t)-QV^n(x,t,\cdot)(i).
\end{split}
\end{equation}
Without loss of generality, we may assume that $V^n(x_0,t_0,i)=\phi(x_0,t_0)$ which implies that $V^n(x,t,i)\leq \phi(x,t)$ in $N(x_0,t_0)$.

Below, we follow along an argument in \cite{Bianetal}.
Let $h>0$ be small enough so that $[t_0,t_0+h)\subset [0,T)$.
Let $\theta=(\theta_u)_{t_0\leq u\leq T}$ be an $\frac{\epsilon}{2}h-$optimal control in $\mathcal{A}^n(x_0,t_0,i)$, and
let $X$ be the unique solution of (\ref{eq:wealth_dyn}) with $X_{t_0}=x_0$ using this $\theta$. Also define the stopping time
\[
\tau:=(t_0+h)\wedge \inf\{s\geq t_0:\,(X_s,s)\notin N(x_0,t_0)\}\wedge \inf\{s\geq t_0:\,Y_s\neq Y_{t_0}\}.
\]
The Dynamic Programming Principle implies that
\[
V^n(x_0,t_0,i) -\frac{\epsilon}{2}h \leq \mathbb{E}_{x_0,t_0,i}[V^n(X_{\tau},\tau,Y_{\tau})]
\]
which combined with Dynkin's formula (which holds because the control $\theta$ is bounded) 
yields
{\small
\[
\begin{split}
&\phi(x_0,t_0)-\frac{\epsilon}{2}h \leq \mathbb{E}_{x_0,t_0,i}[\phi(X_{\tau},\tau)] =\phi(x_0,t_0) \\
& +\mathbb{E}_{x_0,t_0,i}\left[\int_{t_0}^{\tau}\left\{ \frac{\theta_u^2 \sigma^2(i)}{2}\phi_{xx}+[\theta_u(\mu(i)-r(i))+r(i)X_u]\phi_x +\phi_t+QV^n(X_u,u,\cdot)(i)\right\}du\right].
\end{split}
\]}
Using the inequality in (\ref{eq:eps_subopt}), we imply that $-\frac{\epsilon}{2}h \leq  -\epsilon\,\mathbb{E}_{x_0,t_0,i}[\tau-t_0]$ and upon dividing by $h$ we obtain
\[
-\frac{\epsilon}{2} +\epsilon\,\frac{\mathbb{E}_{x_0,t_0,i}[\tau-t_0]}{h} \leq 0.
\]
Taking limit as $h\downarrow 0$, it can be seen that $\mathbb{E}_{x_0,t_0,i}[\tau-t_0]/h\rightarrow 1$ (see \cite{Bianetal}).
This in turn implies that $\epsilon\leq 0$ which is a contradiction.
Thus (\ref{eq:subsolHJB}) holds true and $V^n$ is a subsolution of $F_n(x,t,i,v,v_t,v_x,v_{xx})=0$ on $\bar{\mathcal{O}}$ as desired.

To conclude, observe that $V^n$ increases with $n$ and $V^n\leq V$.
On the other hand, for any $\epsilon$-optimal control $\theta_\epsilon \in \mathcal{A}(x,t,i)$, $\theta_\epsilon \wedge n \in \mathcal{A}^n(x,t,i)$ and using that the utility function $U$ is bounded from below, together with the linearity of the control $\theta$ in the dynamics of $X$,
Fatou's lemma implies that
$\liminf_{n\rightarrow \infty} J(x,t,i;\theta_\epsilon\wedge n) \geq J(x,t,i; \theta_\epsilon)$. These assertions yield
\[
V^n(x,t,i)\leq V(x,t,i) \leq J(x,t,i;\theta_\epsilon) +\epsilon \leq J(x,t,i;\theta_\epsilon\wedge n) +\epsilon  \leq V^n(x,t,i) + \epsilon.
\]
Therefore, $V^n(\cdot,\cdot,i)$ converges to $V(\cdot,\cdot,i)$ pointwise on $\bar{\mathcal{O}}$,
and given that $V(\cdot,\cdot,i)$ is continuous, the locally uniform convergence holds.
\end{proof}


\subsection{Uniqueness}

In the rest of this section, we assert that the value function is the unique constrained viscosity solution of (\ref{eq:viscosityHJB}) on $\bar{\mathcal{O}}$
in the class of concave functions of $x$ and satisfying the boundary condition $V(x,T,i)=U(x)$, $x\in[0,\infty)$.

Roughly, when dealing with (unconstrained) viscosity solutions in an open set, say $\mathcal{O}$,
the classical approach to the uniqueness result is based on the maximum principle, which examines the maximum of $v-w$ on $\mathcal{O}$ where
$v$ is a viscosity subsolution in $\mathcal{O}$ and $w$ is a viscosity supersolution in $\mathcal{O}$.
Indeed, if $v\leq w$ in $\partial{\mathcal{O}}$ then $\sup_{\mathcal{O}}(v-w)\leq 0$ (see e.g. \cite{Fleming&Soner}).
Thus, the value function is the unique viscosity solution in $\mathcal{O}$ with specified boundary conditions on $\partial{\mathcal{O}}$.
The main difficulties when trying to apply this classical comparison result are twofold:
(i) the control set is unbounded, and (ii) we do not know a priori the behavior of the value function on the entire boundary of the domain.


We state comparison and smoothness results, Proposition \ref{prop:uniqueness} and Proposition \ref{prop:smooth} respectively,
without proof and refer to \cite{Duffie&Zarip} and \cite{Zarip94} for the technical details.
The arguments in the proofs of Theorem 4.2 in \cite{Duffie&Zarip} and Theorem 4.1 and 5.1 in \cite{Zarip94} can be adapted to our context after careful consideration.
The main structural difference of the HJB equations lies first on the extra terms accounting for the jumps of the Markov chain
(which can be easily handled because $V$ is bounded as a function of $i$ and continuous for each $i$),
and second, in the presence of the time parameter which can be incorporated in their context as an additional state parameter.
They study an infinite horizon investment \textit{and} consumption problem, and this adds some terms to their HJB equation accounting for the dynamics of the consumption process.
The proofs in the present context are thus very similar and lengthy as the arguments rely mostly on the viscosity property of the solutions and analytical properties of the value function, so we refer the reader to \cite{Duffie&Zarip} and \cite{Zarip94} for the detailed arguments.

\begin{proposition}\label{prop:uniqueness}
Let $u,v:\bar{\mathcal{O}}\times \mathcal{M}\rightarrow \mathbb{R}$ be such that for each $i\in \mathcal{M}$,
\begin{itemize}
    \item[(i)] $u(\cdot,\cdot,i)$ is a viscosity subsolution of (\ref{eq:viscosityHJB}) on $\bar{\mathcal{O}}$, and a concave and upper-semicontinuous function in the first parameter,
    \item[(ii)] $v(\cdot,\cdot,i)$ is a viscosity supersolution of (\ref{eq:viscosityHJB}) in $\mathcal{O}$, bounded from below, uniformly continuous on $\bar{\mathcal{O}}$, and locally H\"older continuous in $\mathcal{O}$.
\end{itemize}
Also assume that for some locally bounded $D,E:[0,T]\rightarrow [0,\infty)$
\begin{equation}\label{eq:sublinear}
|v(x,t,i)| \leq D(t)+E(t)x^\gamma.
\end{equation}
Then $u(\cdot,\cdot,i)\leq v(\cdot,\cdot,i)$ on $\bar{\mathcal{O}}$.
\end{proposition}

Remark that the uniform continuity of the value function $V$ on $\bar{\mathcal{O}}$ follows in view of Propositions \ref{prop:concave} and \ref{prop:t_cont}.


We should comment on the conditions in the last proposition.
The continuity assumption of the viscosity subsolutions and supersolutions is stronger than necessary, but make the presentation somehow simpler. Indeed,
if a viscosity solution is continuous hence it is both upper- and lower-semicontinuous.
Then, in Definition \ref{def:viscosity}, we may allow the viscosity supersolutions (subsolutions, resp.) to be only lower(upper, resp.)-semicontinuous.

To simplify notation, in accordance with \cite{Duffie&Zarip} and \cite{Zarip94}, we write $y=(x,t)\in \bar{\mathcal{O}}$. Fix $i\in \mathcal{M}$.
By contradiction, suppose that
\begin{equation}\label{eq:tocontradict}
\sup_{y\in \bar{\mathcal{O}}} [u(y,i)-v(y,i)]>0.
\end{equation}
Then for $c>0$ small enough and $\lambda\in (\gamma,1)$,
\[
\sup_{(x,t)\in \bar{\mathcal{O}}} [u(y,i)-v(y,i)-c(x+t)^\lambda]>0.
\]
The condition in (\ref{eq:sublinear}) together with the fact that $v$ is bounded from below, imply that the above supremum is attained,  
\[
\sup_{y\in \bar{\mathcal{O}}} [u(y,i)-v(y,i)-c(x+t)^\lambda]=u(\bar{y},i)-v(\bar{y},i)-c(\bar{x}+\bar{t})^\lambda>0.
\]
Next, applying the idea of doubling the variables, for $\delta>0$ small and $\eta\in\mathbb{R}^2$, 
define the function $\phi:\bar{\mathcal{O}}^2\rightarrow \mathbb{R}$, for $y=(x,t),z=(x',t')$, by
\[
\phi(y,z):=u(y,i)-v(z,i)-\left| \frac{z-y}{\delta}-4\eta\right|^4 -c(x+t)^\lambda.
\]
With this notation, following the proof along the lines in \cite{Duffie&Zarip} and \cite{Zarip94}, upon sending $\delta \downarrow 0$,  $c\downarrow 0$ and $\eta \downarrow (0,0)$, it can be seen that (\ref{eq:tocontradict}) is contradicted.

If $w$ is a constrained viscosity solution of (\ref{eq:viscosityHJB}) on $\bar{\mathcal{O}}$ then a standard comparison theorem of viscosity subsolutions and supersolutions yield that
$V\leq w$ in $\mathcal{O}\times \mathcal{M}$. If moreover, $w$ has the same boundary conditions as $V$, $V\leq w$ on $\bar{\mathcal{O}}\times \mathcal{M}$.
Now further assume that $w$ is concave in $x$, then the above proposition along with the regularity properties of the value function yield that $w\leq V$ on $\bar{\mathcal{O}}\times \mathcal{M}$, and the following corollary is immediate.

\begin{corollary}
The value function $V$ in (\ref{OriginalProblem}) is the unique constrained viscosity solution of the HJB equation (\ref{eq:viscosityHJB}) on $\bar{\mathcal{O}}\times \mathcal{M}$ in the class of concave functions in $x$, with the terminal condition $V(x,T,i)=U(x)$ for $x\in[0,\infty)$.
\end{corollary}

\begin{proposition}\label{prop:smooth}
The value function $V$ in (\ref{OriginalProblem}) is the unique function in the class
$C([0,\infty)\times [0,T])\cap C^{2,1}((0,\infty)\times [0,T))$ relative to the parameters $(x,t)$, and concave in $x$,
which satisfies the HJB equation
\begin{equation} \label{eq:HJB}
\sup_{\theta \in \mathbb{R}} \mathbb{L}^\theta V(x,t,i)  =0,
\end{equation}
with terminal condition $V(x,T,i)=U(x)$ for $x\in[0,\infty)$.
\end{proposition}

The proof goes along the lines of Theorem 5.1 with $f= +\infty$ in \cite{Zarip94} and some natural appropriate modifications which we now state.
We need to consider the time parameter.
Zariphopoulou shows that the value function $v$ in \cite{Zarip94}
solves a uniformly elliptic HJB equation in intervals $(x_1,x_2)\subset[0,\infty)$ with boundary conditions $v(x_1)$, $v(x_2)$,
which along with the uniqueness of viscosity solutions yields that $v$ is smooth in $(x_1,x_2)$.
In our context, it can be shown instead that $V$ solves a uniformly parabolic HJB equation in open rectangles $R=(x_1,x_2)\times (t_1,t_2) \subset [0,\infty)\times [0,T]$
with boundary conditions
\[
\begin{aligned}
V(x,t_2,i), & \qquad x\in(x_1,x_2) \\
V(x,t,i),   & \qquad (x,t)\in \{x_1,x_2\}\times [t_1,t_2)
\end{aligned}
\]
Upon using a localization argument similar to that of Proposition \ref{prop:uniqueness} (freeze the Markov chain up to the first jump time),
it is implied that $V(\cdot,\cdot,i)$ is smooth in $R$ (see e.g. Theorem 6.3.6 in \cite{Friedman}).

We conclude this section with the following important result.

\begin{theorem}\label{thm:main}
The value function is given by $V(x,t,i)=U(x)g(i,T-t)$ where $g(i,T-t)$ is in (\ref{eq:F-Krep}).
Moreover, the feedback optimal trading strategy is given by $\theta^*_s=\bar{\theta}(X_s^*,s,Y_s)$ with
\[
\bar{\theta}(x,t,i)=\frac{[\mu(i)-r(i)]x}{(1-\gamma)\,\sigma^2(i)},
\]
where $X^*$ is the optimal wealth process solving (\ref{eq:wealth_dyn}) with $\theta^*$.
\end{theorem}

Observe that the optimal proportion of wealth invested in the risky asset, that is $\theta_s^*/X^*_s$, remains constant during every regime,
which agrees with Merton's proportion when there is no change of regime in the market.

\pagebreak
It is easy to verify that the candidate function $V(x,t,i)=U(x)g(i,T-t)$, as constructed in Section \ref{Construction},
satisfies all the requirements of Proposition \ref{prop:smooth}.
The fact that $V(\cdot,\cdot,i)\in C([0,\infty)\times [0,T])$ and $V(\cdot,\cdot,i)\in C^{2,1}((0,\infty)\times [0,T))$,
follows from $U\in C^2((0,\infty))$ and the Feynman-Kac representation of $g(i,T-t)$ in (\ref{eq:F-Krep}).
Also, the utility function $U$ is 
strictly concave in $x$, so is the candidate function $V$.
The fact that $U(x)g(i,T-t)$ satisfies the HJB equation follows by construction.
Moreover, since $V$ is concave, the trading strategy $\theta^*$ in (\ref{eq:candidatestrat}) is valid. 
Indeed, substitution of
\[
\theta_s^*=\frac{[\mu(Y_s)-r(Y_s)]X_s^*}{(1-\gamma)\,\sigma^2(Y_s)}
\]
in (\ref{eq:wealth_dyn}) gives
\[
dX^*_s=X_s^*\left[\frac{\Phi^2(Y_s)}{(1-\gamma)}+r(Y_s)\right]ds+X_s^* \frac{\Phi(Y_s)}{(1-\gamma)} dB_s, \qquad X^*_t=x,
\]
where $\Phi(y)=[\mu(y)-r(y)]/\sigma(y)$ is the Sharpe ratio. Then, under the strategy $\theta^*$,
wealth stays positive if $x>0$ and it is absorbed at zero if $x=0$. 

It remains to show that $\theta^*$ is an admissible strategy. To see this,
we apply a change of measure with Radon-Nikodym derivative
\[
\frac{d\tilde{\mathbb{P}}}{d\mathbb{P}}=\exp\left\{2\int_t^T \frac{\Phi(Y_u)}{a}dB_u-\int_t^T \frac{\Phi^2(Y_u)}{a^2}du\right\}
\]
and observe that
\[
\mathbb{E}[(X^*_s)^2 \mid \mathcal{F}_t ]=x^2 \tilde{\mathbb{E}}\left[e^{2\int_t^s r(Y_u)\,du }\mid \mathcal{F}_t\right].
\]
Henceforth, for any $s\in [t,T]$,
\[
\mathbb{E}[(\theta_s^*)^2 \mid \mathcal{F}_t]\leq K \,e^{2(s-t)\bar{r}}
\]
where $\bar{r}=\max_{i\in\mathcal{M}} r(i)$ and $K=(\frac{x}{a}\max_{i\in\mathcal{M}}\Phi(i))^2$,
which together with Fubini's Theorem implies that the trading strategy $\theta^*$ satisfies the integrability condition (\ref{eq:admissibility_cond}) and is admissible.


\section{Computation of the value function by Laplace transforms} \label{sec:Laplace}

Consider the power utility $U(x)=\frac{1}{\gamma}x^{\gamma}$, $\gamma<1$, $\gamma\neq 0$.
When there is no regime switching, the optimization problem in (\ref{OriginalProblem}) corresponds to the classical Merton's problem (\cite{Merton69,Merton71})
and it is well-known that the solution is given by
\begin{equation} \label{eq:Merton_value}
V(x,t)=\frac{1}{\gamma}x^{\gamma}  e^{\left[\frac{\gamma}{1-\gamma}\frac{(\mu-r)^2}{2\sigma^2}+\gamma r  \right](T-t)}
\end{equation}
with
\[
\theta^*_s=\frac{(\mu-r)X^*_s}{(1-\gamma)\sigma^2}.
\]
In the regime switching case, we found that the solution is given by
\begin{equation} \label{eq:power_value}
V(x,t,i)=\frac{1}{\gamma}x^{\gamma} g(i,T-t)
\end{equation}
where $g(i,T-t)$ is in (\ref{eq:F-Krep})
with $g(i,0)=1$, and $\theta^*_s$ is as in Theorem \ref{thm:main}.
We next show how to explicitly compute $g(i,T-t)$ by Laplace transform methods.

Given a continuous function $F(t)$ of exponential order (that is, $|F(t)|\leq Ke^{pt}$ for some $K>0$ and $p\geq 0$), the Laplace transform of $F$ is the function defined by
\[
\mathcal{L}\{F(t)\}(u)= \int_0^\infty F(t)e^{-ut}dt.
\]
In what follows, we use two basic properties of the Laplace transform (c.f. \cite{Dyke}) to compute the value function explicitly. Namely,
\[
\mathcal{L}\{F(t)\}(u)=\frac{n!}{(u+b)^{n+1}},\qquad \mbox{if $F(t)=t^n e^{-bt}$},
\]
and the Laplace transform of the convolution of two functions
\[
\mathcal{L}\left\{\int_0^t F(s)G(t-s)ds\right\}(u)=\mathcal{L}\{F(t)\}\mathcal{L}\{G(t)\}(u).
\]

We now concentrate on the computation of the functions $g(i,T-t)$ in (\ref{eq:F-Krep})
for each $i=1,2,\ldots,m$ which solve the coupled PDE
\[
\begin{aligned}
g_t-Qg(\cdot,T-t)(i)+\delta(i) g  & =0 \\
                                                g(i,0)& =1
\end{aligned}
\]
using Laplace transforms and the method can be applied to any set of parameters $\{\delta(i)\in\mathbb{R}: i=1,2,\ldots,m\}$.
We know that the stochastic representation of the solution is given by
\[
g(i,T-t)=\mathbb{E}\left[\exp\left\{-\int_0^{T-t} \delta(Y_u) du \right\} \mid Y_0=i\right].
\]
Remark that we do not require $\delta(i)$ to have a particular sign for the stochastic representation to be valid, because
$\delta(Y_u)$ is piecewise constant, hence uniformly bounded.
We are interested, in particular,
on the case
\[
\delta(i)=-\gamma\left[\frac{(\mu(i)-r(i))^2}{2(1-\gamma)\sigma^2(i)}+ r(i)\right]
\]
which arises in the optimization problem for an investor with a power utility function.

Conditional on the initial condition $Y_0=i$, consider the first jump time of the Markov chain from the state $i$, $\tau_i:=\inf\{t\geq 0: Y_t\neq i\}$.
$\tau_i$ is an exponentially distributed random variable with parameter $q_i$.
Then, by splitting the expectation into the events $\{\tau_i>T-t\}$ and $\{\tau_i\leq T-t\}$, it is easy to see that
\[
\begin{aligned}
g(i,T-t)    & = e^{-(\delta(i)+q_i)(T-t)} \\
            & \hspace{-1.2cm}+\sum_{j\neq i} q_{ij} \int_0^{T-t} e^{-(\delta(i)+q_i)s} \mathbb{E}\left[\exp\left\{-\int_s^{T-t}\delta(Y_u)du\right\} \mid \tau_i=s, Y_s=j\right]ds
\end{aligned}
\]
and by the Markov property, this simplifies to
\begin{equation} \label{eq:system_gs}
g(i,T-t)    = e^{-(\delta(i)+q_i)(T-t)}
 + \sum_{j\neq i} q_{ij}\int_0^{T-t} e^{-(\delta(i)+q_i)s} g(j,T-t-s)ds.
\end{equation}

Taking Laplace transforms for each $i=1,2,\ldots,m$, we obtain
\[
\mathcal{L}\{g(i,T-t)\}(u)=\frac{1}{u+\delta(i)+q_i}+\sum_{j\neq i}\frac{q_{ij}}{u+\delta(i)+q_i}\mathcal{L}\{g(j,T-t)\}(u)
\]
which gives a linear system of $m$ equations in $m$ unknowns. Setting
\[
D_i:=u+\delta(i)+q_i
\]
the system may be written as
\begin{equation}\label{eq:system}
\begin{bmatrix}
    D_1     & -q_{12} & -q_{13} & \dots  & -q_{1m} \\
    -q_{21} & D_2     & -q_{23} & \dots  & -q_{2m} \\
    -q_{31} & -q_{32} & D_3     & \dots  & -q_{3m}\\
    \vdots  & \vdots  & \vdots  & \ddots & \vdots \\
    -q_{m1} & -q_{m2} & -q_{m3} & \dots  & D_m
\end{bmatrix}
\begin{bmatrix}
\mathcal{L}\{g(1,T-t)\}(u) \\
\mathcal{L}\{g(2,T-t)\}(u) \\
\mathcal{L}\{g(3,T-t)\}(u) \\
\vdots\\
\mathcal{L}\{g(m,T-t)\}(u)
\end{bmatrix}
=
\begin{bmatrix}
1\\ 1\\ 1\\ \vdots\\1
\end{bmatrix}.
\end{equation}
The solution of this linear system gives Laplace transforms in the form of proper rational functions in $u$.
Taking inverse Laplace transforms gives the desired solution of the function $g(i,T-t)$.

For instance,
in the case $m=2$, upon using that $q_1\equiv q_{12}$ and $q_2\equiv q_{21}$ we obtain
\[
\mathcal{L}\{g(1,T-t)\}(u)=\frac{1}{D_1}\left[q_1\left(\frac{D_1+q_2}{D_1D_2-q_1q_2}\right)+1\right]
\]
and
\[
\mathcal{L}\{g(2,T-t)\}(u)=\frac{D_1+q_2}{D_1D_2-q_1q_2}.
\]

Substitution of the parameters yields, for $i=1,2$, the proper rational function
\[
\mathcal{L}\{g(i,T-t)\}(u)  = \frac{u+\alpha_i}{u^2+ \beta_1 u + \beta_0},
\]
where
\[
\begin{aligned}
\alpha_1&:=\delta(2)+q_1+q_2 \\
\alpha_2&:=\delta(1)+q_1+q_2 \\
\beta_0&:=\delta(1)\delta(2)+\delta(1)q_2+\delta(2)q_1\\
\beta_1&:=\delta(1)+\delta(2)+q_1+q_2.
\end{aligned}
\]

Let $u_1,u_2$ be the roots of the quadratic polynomial $u^2+ \beta_1 u + \beta_0=0$.
Simple calculations show that $\beta^2_1-4\beta_0>0$, so that there are two distinct real roots.
Then the inverse Laplace transform is, for $i=1,2$,
\[
g(i,T-t)=\frac{1}{u_1-u_2}\left[(u_1+\alpha_i)e^{u_1(T-t)}-(u_2+\alpha_i)e^{u_2(T-t)}\right].
\]
Note that this expression agrees with the terminal condition $g(i,0)=1$.



In general, we imply the following result.

\begin{theorem}
For each $i=1,2,\ldots,m$, the Laplace transforms of $g(i,T-t)$ are in the form
\begin{equation} \label{eq:LaplaceProperFunction}
\mathcal{L}\{g(i,T-t)\}(u)=\frac{u^{m-1}+\alpha_{i,m-2}u^{m-2}+\cdots+\alpha_{i,0}}{u^m+\beta_{m-1}u^{m-1}+\cdots +\beta_0}.
\end{equation}
\end{theorem}

The strictly proper rational function form of the Laplace transform is a direct consequence of Cramer's rule.
The denominator of (\ref{eq:LaplaceProperFunction}) coincides with the determinant of the coefficient matrix on the left of (\ref{eq:system}).
The problem of finding the functions $g(i,T-t)$ for each $i=1,2,\ldots,m$ then reduces to 
decomposing the rational function in (\ref{eq:LaplaceProperFunction}) using partial fractions, and apply the inverse Laplace transform.

The Laplace transform method presented in this section provides an alternative to numerically solving the associated coupled linear system of ordinary differential equations,
which is the classical approach. Moreover, this idea may be exploited in other applications involving hybrid systems.

It is important to remark that in view of the simple functions involved in (\ref{eq:system_gs}), we only needed basic properties of the Laplace transform in order to undertake the method, namely, the linearity and that the Laplace transform of the convolution is the product of the Laplace transforms. These properties also hold for other transforms, such as the Fourier transform and the Z transform (a discrete version of the Laplace transform, which may be used to solve coupled difference equations). Thus, other transforms can also be used to compute the value function, as long as the transforms of the relevant functions are well-defined.

\section{Numerical examples} \label{sec:numerical}

Consider a two-states Markov chain and the following parameters: the term is $T-t=0.5$, and
\[
\begin{aligned}
q_1=20, & \quad \mu(1)=0.5, & \quad \sigma(1)=0.3, & \quad r(1)=0.05,\\
q_2=30, & \quad \mu(2)=0.1, & \quad \sigma(2)=0.5, & \quad r(2)=0.05.
\end{aligned}
\]
We think of state 1 of the Markov chain as modeling a bull market regime, and state 2 a bear market regime.
The parameters reflect empirical studies on bull and bear markets phenomena:
bull markets tend to last longer than bear markets, so we make $q_1$ smaller than $q_2$.
Also, stock returns are lower and volatility is higher during the bear market (c.f. \cite{Gonzalez_etal}).

Table 1 reports the value of the functions $g(1,T-t)$ and $g(2,T-t)$ under different risk-aversion parameters.
The expected utility of terminal wealth increases with the level of risk-aversion, and it is higher if the trading starts in the bull market regime.

\begin{table}[ht] \label{table:2states_gamma}
\caption{Values of $g(i,T-t)$ in a model with two regimes under different $\gamma$.}
\centering
\begin{tabular}{ |c|c|c|c|c|c| }
\hline
$ \gamma $ & \multicolumn{1}{p{2cm}|}{\centering  $ g(1,T-t) $} & \multicolumn{1}{p{2.5cm}|}{\centering  $ g(2,T-t) $}  \\ \hline
0.1 & 1.0419994 & 1.03940982   \\ [0.5ex]
0.3 & 1.1699096 &  1.1587431 \\ [0.5ex]
0.5 & 1.4372864 & 1.40552191  \\ [0.5ex]
0.9 & 28.9779109 &  23.8092044\\ [1ex]
\hline
\end{tabular}
\end{table}

We compare the behavior of the value function as the rate of leaving the bull regime approaches zero in Table 2.
The exponential factor in Merton's solution (\ref{eq:Merton_value}) with $\gamma=0.1, \mu=\mu(1), \sigma=\sigma(1), r=r(1)$ is given by
\[
e^{\left[\frac{\gamma}{1-\gamma}\frac{(\mu-r)^2}{2\sigma^2}+\gamma r  \right](T-t)}\approx 1.067159
\]
Intuitively, if $q_1$ is very close to zero then the Markov chain stays longer in the bull regime and $g(1,T-t)$ is close to Merton's exponential factor.
The numerical experiment is presented in Table 2.

\begin{table}[ht] \label{table:2states_q}
\caption{Behavior of $g(i,T-t)$ as $q_1$ tends to zero.}
\centering
\begin{tabular}{ |c|c|c|c|c|c| }
\hline
$ q_1 $ & \multicolumn{1}{p{2cm}|}{\centering  $ g(1,T-t) $} & \multicolumn{1}{p{2.5cm}|}{\centering  $ g(2,T-t) $}  \\ \hline
20 & 1.0419994 	  & 1.0394098       \\  [0.5ex]
10 & 1.0515394    & 1.0482755    \\  [0.5ex]
1  & 1.0651643    & 1.060905       \\  [0.5ex]
0.1 & 1.0669539   & 1.0625608     \\  [0.5ex]
0.001 &  1.067157 &  1.062749      \\
\hline
\end{tabular}
\end{table}

\pagebreak
Finally, Table 3 reports the value of the functions $g(1,T-t)$, $g(2,T-t)$ and $g(3,T-t)$ under different risk-aversion parameters
under a three-states Markov chain model and assuming the following parameters: the term is $T-t=0.5$, and
\[
\begin{aligned}
q_1=20, & \quad q_{12}=1,  & q_{13}=19, \quad \mu(1)=0.5, & \quad \sigma(1)=0.3, & \quad r(1)=0.05,\\
q_2=30, & \quad q_{21}=25, & q_{23}=5,  \quad \mu(2)=0.1, & \quad \sigma(2)=0.5, & \quad r(2)=0.05, \\
q_3=10, & \quad q_{31}=2,  & q_{32}=8,  \quad \mu(3)=0.3, & \quad \sigma(3)=0.7, & \quad r(3)=0.05.
\end{aligned}
\]

\begin{table}[ht] \label{table:3states_gamma}
\caption{Values of $g(i,T-t)$ with three regimes under different $\gamma$.}
\centering
\begin{tabular}{ |c|c|c|c|c|c|c| }
\hline
$ \gamma $ & \multicolumn{1}{p{2cm}|}{\centering  $ g(1,T-t) $} & \multicolumn{1}{p{2.5cm}|}{\centering  $ g(2,T-t) $} & \multicolumn{1}{p{2.5cm}|}{\centering  $ g(3,T-t) $} \\ \hline
0.1 &    1.0241558      & 1.0221445      & 1.0195109   \\ [0.5ex]
0.3 &    1.0946191      & 1.08632994     & 1.0755275   \\ [0.5ex]
0.5 &    1.231227       & 1.20948267     & 1.1813947   \\ [0.5ex]
0.9 &    7.9600128      & 6.71060261     & 5.31815998   \\ [0.5ex]
\hline
\end{tabular}
\end{table}







\section{Conclusion}

In this paper, we considered an agent who aims to maximize his utility of wealth
by optimally investing on a bond and a risky asset with stochastic regime-switching dynamics.
Assuming that the agent's risk preference is modeled by a power utility function,
we derived the optimal trading strategy by solving the associated Hamilton-Jacobi-Bellman equation which is of degenerate parabolic type,
and showed that the value function is given explicitly in terms of the inverse Laplace transforms of the solution of a given linear system which depends on the parameters of the problem. This method may be applied to other optimization problems involving regime-switching, and other transforms having the same basic properties such as linearity and the transform of the convolution equal the product of the transforms can definitely be applied.
The key feature of the power utility is that we can reduce the number of state variables in the associated PDE by one, which allows us to write the explicit solution for the value function by means of the Feynman-Kac Theorem. Such representation facilitates the computation using Laplace transforms.


\section*{Acknowledgement}

The author would like to thank the referees for all their valuable comments which have tremendously helped to
improve the paper in many aspects.

\end{document}